\documentclass[a4paper]{amsart}

\usepackage{amsmath}
\usepackage{amssymb}
\usepackage{amsthm}
\usepackage{mathrsfs}	

\usepackage{ifthen}

\makeatletter
\renewcommand*\refstepcounter[1]{\stepcounter{#1}%
  \protected@edef\@currentlabel{%
    \csname p@#1\expandafter\endcsname
      \csname the#1\endcsname
  }%
}
\makeatother

\newcommand{\eps}{\varepsilon}

\newcommand{\spinmat}[1]{
   \ifthenelse{\equal{#1}{1}}{\begin{pmatrix}0 & 1    \\ 1   & 0 \end{pmatrix}}{}
   \ifthenelse{\equal{#1}{2}}{\begin{pmatrix}0 & -\im \\ \im & 0 \end{pmatrix}}{}
   \ifthenelse{\equal{#1}{3}}{\begin{pmatrix}1 & 0    \\ 0   & -1\end{pmatrix}}{}
}
\newcommand{\tspinmat}[1]{
   \ifthenelse{\equal{#1}{1}}{\tm{0}{1}{\vphantom{-1}1}{0}}{}
   \ifthenelse{\equal{#1}{2}}{\tm{0}{-\im}{\im}{0}}{}
   \ifthenelse{\equal{#1}{3}}{\tm{1}{0}{0}{-1}}{} 
}
\newcommand{\gamat}[1]{
   \ifthenelse{\equal{#1}{5}}{\begin{pmatrix}0 & \id_2    \\ \id_2   & 0 \end{pmatrix}}{}
}
\newcommand{\tgamat}[1]{
   \ifthenelse{\equal{#1}{5}}{\tm{0}{\id_2}{\vphantom{-1}\id_2}{0}}{}
}
   
\newcommand{\mD}{\mathcal D}

\newcommand{\mL}{\mathcal L}

\newcommand{\wu}{\widetilde u}

\newcommand{\ra}{\alpha} 
\newcommand{\rb}{\beta} 
\newcommand{\rsa}{a} 
\newcommand{\rsb}{b} 

\newcommand{\fa}{\mathfrak a}	
\newcommand{\fh}{\mathfrak h}	
\newcommand{\ft}{\mathfrak t}	
\newcommand{\fv}{\mathfrak v}	

\newcommand{\CH}{\mathcal H} 

\newcommand{\C}{\mathbb C}
\newcommand{\N}{\mathbb N}
\newcommand{\R}{\mathbb R}

\newcommand{\hs}{\mathscr H}

\newcommand{\Ltwo}{\mathscr L^2}

\newcommand{\Ltwoloc}{\mathscr L^2_{\rm loc}}
\newcommand{\Lfourloc}{\mathscr L^4_{\rm loc}}

\newcommand{\Cinfty}{\mathscr C_0^\infty}

\newcommand{\scalar}[2]{\big(#1,#2\big)}

\newcommand{\spec}{\sigma}
\newcommand{\pointspec}{\sigma_{\rm p}}
\newcommand{\essspec}{\sigma_{\rm ess}}

\newcounter{listcounter}
\makeatletter
\newcommand{\CondItem}[1]{
   \renewcommand{\p@listcounter}[1]{#1}	
   \item[#1]
   \refstepcounter{listcounter}		
}
\makeatother

\newenvironment{condlist}
{
\begin{list}{counter is \arabic{listcounter}}{
      \usecounter{listcounter}
      \setlength{\rightmargin}{0.0\textwidth}
      \setlength{\leftmargin}{0.12\textwidth}
      \setlength{\itemindent}{0pt}
      \setlength{\labelsep}{3ex}
      \settowidth{\labelwidth}{ (D2.b) }
   } 
} 
{\end{list}}

\newcounter{proofenumi}
   {\begin{list}{(\roman{proofenumi})}{\usecounter{proofenumi}%
      \setlength{\itemindent}{5.475pt}%
      \setlength{\labelsep}{5.475pt}%
      \setlength{\leftmargin}{0pt}%
      }}%
{\end{list}}

\newcounter{numberedcases}

\DeclareMathOperator*{\esssup}{ess\,sup}
\DeclareMathOperator*{\essinf}{ess\,inf}

\newcommand{\rd}{{\rm d}}
\newcommand{\pdiff}[2][]{\frac{\partial #1}{\partial #2}}

\newcommand{\id}{{I}}
\newcommand{\im}{{\rm i}}

\newcommand{\e}{{\rm e}}

\newcommand{\cc}[1]{\overline{#1}}

\newcommand\la{\lambda}

\usepackage[usenames,dvipsnames]{color}
\usepackage[pdfcenterwindow, a4paper, colorlinks=true, pagebackref=false]{hyperref}

\newcommand{\titlename}{Simplicity of extremal eigenvalues of the Klein-Gordon equation}
\hypersetup{
  pdftitle = {\titlename},
  pdfauthor = {M. Koppen, C. Tretter, M. Winklmeier},
  anchorcolor = red,
  linkcolor = Maroon,
  citecolor = OliveGreen,
  urlcolor = Sepia,
} 

\usepackage{fancyhdr}  
\newtheorem{theorem}{Theorem}[section]
\newtheorem{lemma}[theorem]{Lemma}
\newtheorem{proposition}[theorem]{Proposition}

\theoremstyle{definition}
\newtheorem{definition}[theorem]{Definition}
\newtheorem{example}[theorem]{Example}

\theoremstyle{remark}
\newtheorem{remark}[theorem]{Remark}

\newtheorem{corollary}[theorem]{Corollary}

\numberwithin{equation}{section}

\begin{document}


\title[Simplicity of extremal eigenvalues of the Klein-Gordon equation]{Simplicity of extremal eigenvalues of the Klein-Gordon equation}

\author{Mario Koppen}
\address{Mathematics Centre, Technical University of Munich, Boltzmannstr.\ 3, 85748 Garching, Germany}
\email{mario@ma.tum.de}
\author{Christiane Tretter}
\address{Institute of Mathematics, University of Bern, Sidlerstr.\ 5, 3012 Bern, Switzerland}%
\email{tretter@math.unibe.ch}
\author{Monika Winklmeier}
\address{Department of Mathematics, Universidad de Los Andes, Cra.\ 1a No 18A-10, A.A.\ 4976 Bogot\'a, Colombia}
\email{mwinklme@uniandes.edu.co}
\dedicatory{Dedicated to Professor Heinz Langer on the occasion of his 75th birthday}
\begin{abstract}
  We consider the spectral problem associated with the Klein-Gordon equation for unbounded electric potentials. If the spectrum of this problem is contained in 
  two disjoint real intervals and the two inner boundary points are eigenvalues, we show that these extremal eigenvalues are simple and possess 
  strictly positive eigenfunctions. Examples of electric potentials satisfying these assumptions are given.
\end{abstract}
\date{\today} 

\maketitle

\section{Introduction} 
\label{sec:intro}
The motion of a spin-$0$ particle with mass $m>0$ and electric charge $e>0$ in an exterior electromagnetic field in $n$ spatial dimensions is described by
the time de\-pend\-ent Klein-Gordon equation. With the physical units chosen such that $c=\hbar=1$, the Klein-Gordon equation takes the \vspace{-0.8mm} form
\begin{align}
   \label{eq:KGtime}
   \biggl[
      \Bigl( \pdiff{t} - \im\, eq \Bigr)^2 
      + \sum_{j=1}^n \Bigl( -\im\, \pdiff{x_j} - e A_j \Bigr)^2
      + m^2
      \biggr] \wu
   \ =\ 0
\end{align}
in $\Ltwo(\R^n)$ where the electric potential $q$ and the components $A_j$ of the vector~potential are real-valued functions.
For the corresponding Cauchy problem, solutions $\wu(\,\cdot\,, t)\in\Ltwo(\R^n)$, $t\in R$, are subject to an initial condition $\wu(\,\cdot\,, 0)=u_0$ with $u_0\in\Ltwo(\R^n)$.

The solvability of this Cauchy problem is closely related to the spectral problem 
\begin{equation}
\label{eq:KG}
   \Bigl( H_0 - (\lambda - V)^2 \Bigr)u\ =\ 0,
\end{equation}
formally obtained from the Klein-Gordon equation \eqref{eq:KGtime} by means of the ansatz $\wu(\,\cdot\,, t) = \e^{\im\lambda t}u$. 
Here $H_0$ is a self-adjoint realization in $\Ltwo(\R^n)$ of the formal second order  \vspace{-0.8mm} differential~expression
\begin{align}
   \label{eq:H0}
   \sum_{j=1}^n \Bigl( -\im\, \pdiff{x_j} - e A_j \Bigr)^2 + m^2
\end{align}
and $V$ is the  multiplication operator by $eq$ in the Hilbert space $\Ltwo(\R^n)$ with maximal domain, which is in general unbounded.
Note that, in \eqref{eq:KG} and \eqref{eq:H0}, suitable assumptions on $V$ and $A_j$ are required to properly define the sums of operators, at least in the sense of quadratic forms.

For $V=0$, a point $\la$ belongs to the spectrum of the eigenvalue problem \eqref{eq:KG} if and only if $\la^2$ belongs to the spectrum of $H_0$.
Since the latter is continuous and given by $\sigma(H_0)=\sigma_{\rm ess}(H_0)=[m^2,\infty)$, the spectrum of 
\eqref{eq:KG} is continuous and consists of the two disjoint intervals $\big(-\infty,-m\big] \dot\cup \big[m,\infty\big)$. 
For bounded $V$ with $\|V\|< m$, classical perturbation arguments show that the spectrum of \eqref{eq:KG} 
is contained in the two disjoint intervals
\begin{equation}
\label{specgap}
  \big(\!-\!\infty,-m+\|V\|\big] \,\dot\cup\, \big[m-\|V\|,\infty\big).
\end{equation}
For bounded $V$ with $\|V\| > m$,  it was observed already in the 1940ies that \eqref{eq:KG} may have non-real spectrum (see \cite{schiff}); the latter is related to the so-called Klein paradox. 

On the other hand, even for unbounded $V$, the spectrum of \eqref{eq:KG} may remain real and retain a spectral gap (see \cite{LNT06}); e.g.\ if $V$ is $H_0^{\frac 12}$-bounded with $\|VH_0^{-\frac 12}\|<1$, then the spectrum of \eqref{eq:KG} 
is contained in the two disjoint intervals
\begin{equation}
\label{eq:unbdd-gap}
  \big(\!-\!\infty,-m+\|VH_0^{-\frac 12}\|m\big] \,\dot\cup\, \big[m-\|VH_0^{-\frac 12}\|m,\infty\big).
\end{equation}
The aim of this paper is to investigate if, in the case when eigenvalues bound the spectral gap, they are simple. In analogy with the ter\-min\-ology for Schr\"odinger operators, we call these extremal eigenvalues ground states. 

For bounded $V$, the simplicity of the ground states was proved in \cite{Naj83}. For relatively bounded $V$, the eigenvalues in the gap of the essential spectrum were studied and estimated by means of variational principles in \cite{LT06}, but their multiplicities were not investigated. In \cite{Naj83}, as well as in \cite{LNT06}, \cite{LNT08}, the Klein-Gordon equation or, equivalently, the spectral problem \eqref{eq:KG} was linearized to obtain a first order system of differential equations or a spectral problem that depends linearly on the eigenvalue parameter $\la$, respectively.

For the purpose of the present paper, the original form \eqref{eq:KG} of the Klein-Gordon spectral problem used in \cite{LT06} is more advantageous for two reasons. First, the operators $T(\la)$ associated with the left-hand side of \eqref{eq:KG} depend quadratically on the spectral parameter $\la$.
This allows us to use the theory of strongly damped quadratic operator polynomials to ensure the reality of the spectrum due to the existence of a spectral gap and to prove that all the eigenvalues are semi-simple, i.e.\ there are no associated vectors.  Secondly, for fixed $\lambda$ in the spectral gap, the operator $T(\la)$ is a semi-bounded perturbation of the free Schr\"odinger operator $H_0=-\Delta+m^2$. Therefore we may apply a Krein-Rutman type theorem to establish the simplicity of the ground states of the Klein-Gordon spectral problem~\eqref{eq:KG}. 

The theory of self-adjoint quadratic operator polynomials in a Hilbert space~$\CH$, 
$$
L(\la)=\la^2 A + \la B + C, \quad \la \in \C,
$$ 
with $A=I$, bounded $B$, and compact positive $C$, was developed in the seminal~work \cite{MR0195291} of M.G.\ Krein and H.\ Langer (the results of which first appeared in \cite{MR0169060} and were translated in \cite{MR511976}, \cite{MR516767}). For non-compact $C$, this theory was further developed by H.\ Langer in \cite{hhabil} under the assumption that the pencil $L$ is strongly damped. This notion goes back to R.\ Duffin (see \cite{MR0069030}) and means~that 
$$
  (Bx,x)^2 > 4 \, (Ax,x) (Cx,x),  \quad x\in \CH \setminus \{0\};
$$ 
as a consequence, each quadratic equation $(L(\la)x,x)=0$ has two different real solutions $p_-(x) < p_+(x)$ and the two (real) root zones $\Delta_\pm := \{ p_\pm(x) : x \in \CH \setminus \{0\} \}$ are disjoint and may have at most a common boundary point. It is well-known that the spectrum of a strongly damped operator polynomial is real and all Jordan chains have length~$1$ (see e.g.\ \cite[\S 31]{MR971506} where the term `hyperbolic' is used instead of strongly damped). The case of strongly damped operator pencils with unbounded coefficients was studied by A.A.\ Shkalikov (see \cite{MR933487}); here the two root zones $\Delta_\pm$ may touch each other also at $\infty$, as it is the case for the Klein-Gordon pencil~$T$ (see \eqref{eq:op-pol} below).

An operator version of the well-known Perron-Frobenius theorem for matrices with positive entries was established by M.G.\ Krein and M.A.\ Rutman (see \cite{MR0027128}, \cite{MR0038008}). In the original version, if $(M, \rd\mu)$ is a measure space and $T$ is an operator in $\Ltwo(M, \rd\mu)$ such that, for some integer $k$, the operator $T^k$ is positivity improving, i.e.\ $T^k f$ is positive for every non-negative $f\in \Ltwo(M, \rd\mu) \setminus \{0\}$, then the largest eigenvalue in modulus of $T$ is real, positive, and simple with a strictly positive eigenfunction. There are numerous generalizations and extensions of the Krein-Rutman theorem. Here we apply a corresponding result by W.G.\ Faris  for bounded self-adjoint operators in $\Ltwo(M, \rd\mu)$, which applies to the resolvent or the semi-group generated by semi-bounded Schr\"odinger operators (see \cite{faris}).

The paper is organised as follows. 
In Section~\ref{sec:T} we consider an abstract form of the Klein-Gordon spectral problem \eqref{eq:KG} where $H_0$ is assumed to be a uniformly positive operator in a Hilbert space $\CH$ and $V$ is a symmetric perturbation which is $H_0$-form-bounded with $H_0$-form bound $<1$; the latter implies, in particular, that $\mD(H_0^{\frac 12}) \subset \mD(V)$. Under these assumptions, with the right-hand side of \eqref{eq:KG} we can associate a self-adjoint operator 
\begin{equation}
\label{eq:op-pol}
  T(\la) := H_0 \dot - V^2 + 2 \la V - \la^2, \quad \la \in \C,
\end{equation}
where $H_0 \dot - V^2$ denotes the operator form sum of $H_0$ and $-V^2$.
In Section~\ref{sec:reality} we study the quadratic operator polynomial $T$ given by \eqref{eq:op-pol} by means of the associated quadratic form $\ft(\la)$ defined on $\mD(\ft(\la))= \mD(H_0^{\frac 12}) \subset \mD(V)$. If the two root zones $\Lambda_\pm$ consisting of all zeros $p_\pm(x)$~of $\ft(\la)[x]=0$ for $x \in \mD(H_0^{\frac 12})$ are separated by a gap, then $T$~is strongly damped and hence all eigenvalues of $T$ are semi-simple, i.e.\ there are no associated vectors. In Section \ref{sec:condV} we establish explicit conditions on the operator $V$ in \eqref{eq:op-pol} guaranteeing that the two root zones are 
contained in two disjoint intervals separated by a gap, such as in \eqref{eq:unbdd-gap}. In Section \ref{sec:positivity} we specialize to the Klein-Gordon equation \eqref{eq:KG} in $\Ltwo(\R^n)$ where $H_0=-\Delta+m^2$. We show that, for $\la$ in the gap between the two zones $\Lambda_\pm$, the corresponding operator $T(\la)$ is positivity improving so that a Krein-Rutman type theorem applies. This yields that the ground states 
 of the Klein-Gordon equation are simple. Finally, in Section \ref{sec:examples} we show how our results apply to some concrete potentials, including Coulomb-like and Rollnik~potentials with vanishing vector potential $A$. 

The following definitions are used throughout the paper. For a linear operator~$A$ in a Hilbert space $\CH$, we denote by $\mD(A)$ its \emph{domain} and by $\ker A$ its \emph{kernel}. \linebreak A sesquilinear form $\fa$ with domain $\mD(\fa)$ is called \emph{symmetric} if $\fa[x,y] = \overline{\fa[y,x]}$, \linebreak $x,y\in\mD(\fa)$, and we write  
$\fa[x]:=\fa[x,x]$, $x\in\mD(\fa)$, for the associated quadratic~form. 
For a symmetric linear operator $A$, we write $A\gg 0$ if there exists some $\gamma_0 > 0$ such that $(Ax,x) \ge \gamma_0 \|x\|^2$, $x\in \mD(A)$; the same notation is used for quadratic forms.
The \emph{numerical range} of a linear operator $A$ and of a quadratic form~$\fa$, respectively, are given by
\begin{equation}
\label{numrange0}
\begin{array}{ll}
   W(A)\hspace*{-2.5mm}&:= \{ \scalar{Ax}{x}:\, x\in\mD(A),\ \|x\| = 1 \}, \\[1mm]
   W(\fa)\hspace*{-2.5mm}&:= \{ \,\fa[x]:\, x\in\mD(\fa),\ \|x\| = 1 \}. 
\end{array}   
\end{equation}
For a symmetric operator $A$, the numerical range $W(A)$ is real; for a self-adjoint operator $A$, the inclusion $\spec(A) \subset \overline{W(A)}$ holds.
If $\fa$ is a densely defined closed symmetric sesquilinear form bounded from below and $A$ is the self-adoint operator associated with $\fa$ by the first representation theorem, i.e.\ $\fa[x,y]= \scalar{Ax}{y}, \, x\in\mD(A)$, $y\in \mD(\fa)$, then $W(A)$ is dense in $W(\fa)$ (see \cite[Thm.\ VI.2.1, Cor.\ VI.2.3]{kato}). 
More details on linear operators and quadratic forms may be found in~\cite{kato}.


If $T$ is an analytic operator function on $\C$, i.e.\ a function $\la \mapsto T(\la)$ defined on $\C$ whose values $T(\la)$ are closed linear operators,
the \emph{resolvent set}, \emph{spectrum}, \emph{point spectrum} and \emph{essential spectrum} of $T$ are defined by
\begin{align*}
   \rho(T) &:= \{ \lambda\in \C \,:\, 0\in\rho(T(\lambda)) \} = \{ \la \in \C : T(\la) \text{ bijective}\},
   \\
   \spec(T) &:= \{ \lambda\in \C \,:\, 0\in\spec(T(\lambda)) \} = \C \setminus \rho(T),
   \\
   \pointspec(T) &:= \{ \lambda\in \C \,:\, 0\in\pointspec(T(\lambda)) \} = \{ \la \in \C : T(\la) \text{ not injective}\} ,
   \\
   \essspec(T) &:= \{ \lambda\in \C \,:\, 0\in\essspec(T(\lambda)) \} = \{ \la \in \C : T(\la) \text{ not Fredholm}\}.
\end{align*}
If $\lambda_0 \in \C$ is an eigenvalue of the operator function $T$, then a sequence $(x_j)_{j=0}^{m-1}$ is called a Jordan chain of length $m$ of $T$ in $\lambda_0$ if
\begin{align}
\label{Jordan}
   \sum_{k=0}^j \frac{1}{k!}\ T^{(k)}(\lambda_0) x_{j-k}\, =\, 0,
   \quad j=0,\, \dots,\, m-1.
\end{align}
An eigenvalue $\lambda_0$ of $T$ is called \emph{semi-simple} if the maximal length of a Jordan chain is $1$;
it is called \emph{simple} if $\la_0$ is semi-simple and $\dim \ker T(\la_0) = 1$.
If the domains $\mD(T(\la))=: \mD_0$ are independent of $\la$, we define the \emph{numerical range} of $T$ as
\begin{align}
\label{numrange}
  W(T) := \{ \la \in \C : \exists \ x \in \mD_0 \ (T(\la)x,x) = 0 \}.
\end{align}
Clearly $\pointspec(T) \subset W(T)$. For operator functions $T$ with bounded values, the inclusion $\spec(T) \subset \overline{W(T)}$ 
holds if there exists a $z_0\in\C$ such that $0 \notin \overline{W(T(z_0))}$ (see \cite[Thm.\ 26.6]{MR971506}). 

All the above notions are defined analogously for analytic operator functions defined on some domain $\Omega\subset \C$; since the operator functions occurring in this paper are polynomials, we may restrict ourselves to the case $\Omega=\C$. 
More details on analytic operator functions with bounded values may be found in \cite{MR971506}.


\section{The abstract Klein-Gordon pencil $T$} 
\label{sec:T}

Suppose that $H_0$ is an unbounded self-adjoint uniformly positive operator in a Hilbert space $\CH$, $H_0 \ge m^2 > 0$, and $V$ is a symmetric operator in $\CH$.
In this section, we associate a self-adjoint operator $T(\la)$ with the formal operator sum $H_0 - (\lambda - V)^2$ in \eqref{eq:KG} by means of quadratic forms.

To this end, we assume that $V$ satisfies the following two assumptions:
\begin{condlist}
   \CondItem{(V1)}\label{A1} 
   $\mD(H_0^{\frac{1}{2}}) \subseteq \mD(V)$; \vspace{1.5mm}
   \CondItem{(V2)}\label{A3} 
   there exist $\ra$, $\rb\ge 0$ with $\rb <1$ such that
\begin{align}
\label{eq:H1hBoundSquare}
   |\scalar{V x}{Vx}| \le \ra \,\|x\|^2 + \rb\, \big|\scalar{H_0^{\frac 12} x}{H_0^{\frac 12} x}\big|,
   \quad x\in\mD(H_0^{\frac 12}).
\end{align}
\end{condlist}
Assumption \ref{A3} means that the self-adjoint operator $V^\ast V$, of which $V^2$ is a restriction, is $H_0$-form-bounded with $H_0$-form bound $<1$;
the infimum over all $\rb$ such that \eqref{eq:H1hBoundSquare} holds for some $\ra\ge 0$ is called \emph{$H_0$-form bound} of $V^\ast V$ (see \cite[Chapt.~X.2]{RSII}).

\begin{remark}
\label{rem:kato}
It is well-known (see e.g.\ \cite[Sect.\ V.4.1]{kato}) that the existence of $\ra$, $\rb\ge 0$, $\rb <1$, with \eqref{eq:H1hBoundSquare} is equivalent to the existence of $\rsa$, $\rsb\ge 0$, $\rsb <1$, such that 
\begin{alignat}{3}
   \label{eq:H1hBound}
   \| Vx \| &\le \rsa \, \|x \| + \rsb \,\big\|H_0^{\frac{1}{2}} x\big\|,
   &\quad& x\in\mD(H_0^{\frac{1}{2}});
\end{alignat}
in fact, \eqref{eq:H1hBoundSquare} implies \eqref{eq:H1hBound} with $\rsa=\sqrt \ra$, $\rsb=\sqrt \rb$, while \eqref{eq:H1hBound} implies  \eqref{eq:H1hBoundSquare} with $\ra=(1+\eps^{-1}) \rsa^2$, $\rb=(1+\eps) \rsb^2$ for arbitrary $\eps>0$. Hence assumption \ref{A3} is equivalent to the assumption 
\begin{condlist}
 \CondItem{(V2$'$)}\label{A3'} 
$V$ is $H_0^{\frac 12}$-bounded with $H_0^{\frac 12}$-bound $<1$.
\end{condlist}

\end{remark}


Assumption \ref{A1} alone already implies that $VH_0^{-\frac 12}$ is a bounded operator. 
The  norm of~$VH_0^{-\frac 12}$ is related to the constants in assumption \ref{A3} as follows.

\begin{proposition}
\label{prop:norms}
Suppose that assumption {\rm \ref{A1}} holds. If {\rm \ref{A3}} is satisfied with constants $\ra$,~$\rb$ in \eqref{eq:H1hBoundSquare} or with constants $\rsa$, $\rsb$ in \eqref{eq:H1hBound}, respectively,  then 
\begin{align}
\label{eq:normS}
 \|VH_0^{-\frac 12}\| \le \frac{\rsa}m+\rsb, \quad 
 \|VH_0^{-\frac 12}\| \le \frac{\sqrt\ra}m+\sqrt\rb, \quad 
 \|VH_0^{-\frac 12}\| \le \sqrt{\frac{\ra}{m^2}+\rb}.
\end{align}
In particular, the following are equivalent: 
 \begin{enumerate}
  \item[{\rm i)}] $\big\|VH_0^{-\frac 12} \big\| < 1$;
  \item[{\rm ii)}] assumption {\rm \ref{A3}} holds with $\rsa, \rsb \ge 0$ in \eqref{eq:H1hBound} such that $\displaystyle \frac\rsa m + \rsb < 1$;
  \item[{\rm iii)}] assumption {\rm \ref{A3}} holds with $\ra, \rb \ge 0$ in \eqref{eq:H1hBoundSquare} such that $\displaystyle\frac {\sqrt\ra} m + \sqrt\rb < 1$;
  \item[{\rm iv)}] assumption {\rm \ref{A3}} holds with $\ra, \rb \ge 0$ in \eqref{eq:H1hBoundSquare} such that $\displaystyle\sqrt{\frac{\ra}{m^2}+\rb} < 1$.
\end{enumerate}
\end{proposition}

\begin{proof}
Since $H_0 \ge m^2$, we have  $\|H_0^{-\frac 12}\|\le 1/m$. If \eqref{eq:H1hBoundSquare} or \eqref{eq:H1hBound} hold, then the estimates 
\begin{align}
\label{mid1}
  \big\| VH_0^{-\frac 12} x \big\|^2 \!& \le \ra \,\big\|H_0^{-\frac 12} x \big\|^2 + \rb \, \big\| H_0^{\frac 12} H_0^{-\frac 12} x \big\|^2 
  \le \Big(  \frac\ra {m^2} + \rb \Big) \|x\|^2,\\
\label{mid2}  
  \big\| VH_0^{-\frac 12} x \big\| \ & \le \rsa \,\big\|H_0^{-\frac 12} x \big\| + \rsb \, \big\| H_0^{\frac 12} H_0^{-\frac 12} x \big\| 
  \le \Big(  \frac\rsa m + \rsb \Big) \|x\|,
\end{align}
for $x\in \CH$ imply the first and the third estimate in \eqref{eq:normS}. Since \eqref{eq:H1hBoundSquare} with $\ra$, $\rb$ implies \eqref{eq:H1hBound} with  $\rsa=\sqrt\ra$, $\rsb=\sqrt\rb$ by Remark \ref{rem:kato}, the second estimate in \eqref{eq:normS} follows from the first. 
 
i) $\Rightarrow$ ii), i) $\Rightarrow$ iii), i) $\Rightarrow$ iv): \ Since $ VH_0^{-\frac 12}$ is bounded by {\rm \ref{A1}}, the estimate
\[
  \|Vx\| = \big\| VH_0^{-\frac 12} H_0^{\frac 12} x \big\| \le \big\| VH_0^{-\frac 12} \big\| \big\| H_0^{\frac 12} x \big\|, \quad x\in \mD(H_0^{\frac 12}),
\]
shows that \eqref{eq:H1hBoundSquare} holds with $\ra=0$, $\rb= \| VH_0^{-\frac 12} \|^2 < 1$ and that \eqref{eq:H1hBound}
holds with $\rsa=0$, $\rsb= \| VH_0^{-\frac 12} \| < 1$.

ii) $\Rightarrow$ i), iii) $\Rightarrow$ i), iv) $\Rightarrow$ i): All implications are obvious from the estimates in~\eqref{eq:normS}.
\end{proof}

The next lemma shows that if conditions \ref{A1}, \ref{A3} are satisfied, then, for every $\la\in\C$, there is a well-defined self-adjoint operator $T(\la)$ associated with the formal operator sum $H_0 - (\lambda - V)^2$ in the abstract Klein-Gordon spectral problem \eqref{eq:KG}.

\begin{lemma}
   \label{lemma:form}
   Assume that conditions {\rm \ref{A1}}, {\rm \ref{A3}} hold. Then, for every $\lambda\in\C$, 
   there exists a unique closed sectorial operator $T(\la)$ in $\CH$ such that 
   \begin{align*}
      \big( T(\la) x,y \big) =
      \scalar{ H_0^{\frac{1}{2}}x }{ H_0^{\frac{1}{2}}y }
      - \scalar{ (\lambda-V)x }{ (\cc\lambda-V)y} =: \ft(\lambda)[x,y]
   \end{align*}
   for all $x$, $y\in\mD(H_0^{\frac{1}{2}})=: \mD(\ft(\lambda))=:\mD_\ft$. The corresponding operator function  $T$ has the following properties:
   \begin{itemize}
    \item[{\rm i)}]  $T$ is self-adjoint, i.e.\ $T(\la)^*=T(\overline\la)$ for $\la \in\C$, 
    and for $\la \in \R$ the self-adjoint operators $T(\la)$ are bounded from below; \vspace{1mm}
    \item[{\rm ii)}]  the domains $\mD(T(\la))=: \mD_T$ are independent of $\la \in \C$ and
\begin{align}
\label{T-id}
  T(\lambda) = T(\mu) + 2 (\lambda-\mu) (V - \mu) - (\lambda-\mu)^2, \quad \lambda,\mu\in \C;
\end{align}
    \item[{\rm iii)}] $T$ is analytic with derivatives given by
    \begin{align*}
      \hspace{1cm} 
      T'(\lambda)x\, =\, 2(V-\lambda)x,
      \hspace{2ex}
      T''(\lambda)x\, =\, -2 x,
      \hspace{2ex}
      T^{(j)}(\lambda)x\, =\, 0,\ j=3,4, \dots, 
     \end{align*}
   for $x\in\mD(T(\lambda))$.
  \end{itemize}
\end{lemma}

\begin{proof}
i) The operator $H_0$ is positive and self-adjoint and hence the corresponding form given by $\fh_0[x,y]:=\scalar{ H_0^{\frac{1}{2}}x }{ H_0^{\frac{1}{2}}y }$, $x$, $y\in\mD(H_0^{\frac{1}{2}})$, is closed and positive.
The form $\fv(\la)$ given by $\fv(\la)[x,y]:=\scalar{ (\lambda-V)x }{ (\cc\lambda-V)y}$, $x,y \in \mD(V)$, is sectorial, and it is symmetric if $\la\in \R$. 
By \ref{A3} and Remark \ref{rem:kato}, there exist $\rsa$, $\rsb\ge 0$, $\rsb <1$, such that 
\[
   \| (V-\la)x \| \le \|Vx\| + |\la| \, \|x\| \le \big(\rsa + |\la|\big) \, \|x \| + \rsb \,\big\|H_0^{\frac{1}{2}} x\big\|,
   \quad x\in\mD(H_0^{\frac{1}{2}}).
\]
Using Remark \ref{rem:kato} again, we see that, for $x\in\mD(H_0^{\frac{1}{2}})$,
\begin{align*}
   \big|\big((V-\la)x,(V-\overline{\la})x\big)\big| & \le \|(V-\la)x\| \, \|(V-\overline{\la})x\| = \|(V-\la)x\|^2 \\
   & \le  (1+\varepsilon^{-1})(\rsa +|\la|)^2\!\, \|x \|^2 + (1+\varepsilon) \rsb^2 \,\big(H_0^{\frac{1}{2}} x,H_0^{\frac{1}{2}} x \big)
\end{align*}
with arbitrary $\varepsilon >0$.
Choosing $\varepsilon < 1/\rsb^2-1$, we see that the form $\fv(\la)$ is $\fh_0$-bounded with $\fh_0$-bound~$<1$. Hence, 
for $\la\in \C$, the form $\ft(\la)$ is closed and sectorial and the existence of $T(\la)$ follows from the first representation theorem (see e.g.\ \cite[Thm.\ VI.2.1]{kato}). 
Since $\ft(\la) = \ft^*(\overline\la)$ for $\la \in \C$, the operator function $T$ is self-adjoint.
The remaining claims for $\la \in \R$ (and, in fact, also the existence of $T(\la)$ for $\la\in \R$) follow
from the so-called KLMN theorem (due to Kato-Lions-Lax-Milgram-Nelson, see \cite[Thm.\ X.16]{RSII} or \cite[Sect.\ 2.1]{MR2159561}) and from \cite[Thm.\ V.4.11]{kato}. 

ii) It is easy to see that, for  $\mu$, $\la\in \R$ and $x$, $y\in \mD_\ft=\mD(H_0^{\frac{1}{2}})$,
   \begin{align}
      \label{eq:FormExpansion}
      \ft(\lambda)[x,y] = 
      \ft(\mu)[x,y]
      + 2 (\lambda-\mu) \scalar{ (V-\mu)x }{ y }
      - (\lambda-\mu)^2 \scalar{ x }{ y}.
   \end{align}
In particular, the operator $T(\lambda)$ associated with the form on the left hand side coincides with the operator which is associated with the sum of forms on the right hand side; this operator has the same domain as the operator $T(\mu)$ associated with the first term $\ft(\mu)$ because $\scalar{ (V-\mu)x }{ y }$ defines a $\ft(\mu)$-bounded form with $\ft(\mu)$-bound $0$ and $ \scalar{ x }{ y}$ is a bounded form (see \cite[Thm.\ VI.1.33]{kato}).

iii) The formulas for the derivatives of $T$ follow from the identity \eqref{T-id}.
\end{proof}

\begin{corollary}
\label{cor:op-pol}
If we denote by $H_0  \dot - V^2:= T(0)$ the operator form sum corresponding to $\ft(0)$, then
\[
 T(\la) = H_0 \dot - V^2 + 2 \la V - \la^2,  \quad \mD(T(\la)) = \mD_T = \mD(H_0 \dot - V^2),   \quad \la \in \C.
\]
\end{corollary}

\begin{proof}
The claim is immediate from Lemma \ref{lemma:form} ii) with $\mu=0$.
\end{proof}

\begin{remark}
Note that, e.g.\ if the operator $V$ is bounded, then $T(0) = H_0 - V^2$ may be defined as an operator sum; however, the operator form sum $T(0)=H_0  \dot - V^2$ may even be defined if $\mD(H_0) \cap \mD(V^2) = \{0\}$. 

In particular cases, e.g.\ for the Klein-Gordon equation in $\R^3$, it may even be possible to define the operators $T(\la)$ without using assumption \ref{A3} by means of the Leinfelder-Simader theorem \cite[Theorem 4]{LeSi81}.
\end{remark}

By assumption \ref{A1}, the operator $S:= VH_0^{-\frac 12}$ is bounded and hence the quadratic operator polynomial  
$L$ in the Hilbert space~$\CH$, given by
\begin{align}\label{eq:Ldef}
   L(\lambda) 
   = \id - \big(S^* - \lambda H_0^{-\frac{1}{2}}\big) \big(S - \lambda H_0^{-\frac{1}{2}}\big), \quad \la\in \C,
\end{align}
has bounded coefficients. However, the numerical range $W(L)$ of $L$ is not bounded  since $L$ is not monic and its leading coefficient $-H_0^{-1}$ is not bounded away from $0$. 

The following relation between the operator polynomials $T$ and $L$ was proved in \cite{LT06}; a similar factorisation may be found in 
\cite[Equ.\ (4.9)]{Ves83}.

\begin{proposition}
\label{T-L}
Suppose that assumption {\rm \ref{A1}} holds, i.e.\ $\mD(H_0^{\frac{1}{2}}) \subseteq \mD(V)$. Then
\begin{align}
   \label{eq:L}
   T(\lambda)\, =\, H_0^{\frac{1}{2}} L(\lambda) H_0^{\frac{1}{2}},
   \quad \lambda\in\C,
\end{align}
and we have
\begin{enumerate}
\item[{\rm i)}]
$\pointspec(T) \subset W(T) \subset W(L)=W(\ft)$;
\item[{\rm ii)}]
$\spec(T) \subset \spec(L)$, \ $\pointspec(T)=\pointspec(L)$;
\item[{\rm iii)}]
$\spec(T)\cap\R = \spec(L)\cap\R$, \ $\essspec(T)\cap\R = \essspec(L)\cap\R$.
\end{enumerate}
\end{proposition}

\begin{proof}
All claims except for i) were proved in \cite[Prop.\ 2.3]{LT06}. 
The first inclusion in i) is obvious. In fact, if $\la_0 \in \pointspec(T)$, then there exists $x_0 \in \mD_T \setminus \{0\}$ such that $T(\la_0)x_0 = 0$. Taking the scalar product with $x_0$ yields $\la_0 \in W(T)$. The second inclusion in i) follows from identity \eqref{eq:L}; the last equality follows from the~relation
\begin{align}
\label{L-ft}
  \big(L(\la)H_0^{\frac 12}x, H_0^{\frac 12}x\big) = \ft(\la)[x], \quad x \in \mD_\ft=\mD(H_0^{\frac{1}{2}}),
\end{align}
and from the fact that $H_0^{\frac 12}$ is bijective. 
\end{proof}

\begin{remark}
In \cite{LNT06}, \cite{LNT08}, and \cite{LT06}, the abstract Klein-Gordon equation was studied under the assumption
\begin{condlist}
   \CondItem{(V3)}\label{A2} 
   $VH_0^{-\frac{1}{2}} =: S = S_0 + S_1$ where $\|S_0\| < 1$ and $S_1$ is compact,
\end{condlist}
which implies condition \ref{A3}. In fact, since $S_1$ is compact, the operator $S_1 H_0^{\frac{1}{2}}$ has $H_0^{\frac{1}{2}}$-bound $0$ and hence, for $\eps < 1-\|S_0\|$,
   there exists an $\ra\ge 0$  such that $\| S_1 H_0^{\frac 12}x\| \le \ra \| x\| + \eps \|H_0^{\frac 12} x\|$ for $x\in\mD(H_0^{\frac 12})$ and~so
   \begin{align*}
      \| Vx\| 
      = \| S_1 H_0^{\frac 12}x + S_0 H_0^{\frac 12}x\|
      \le \ra \,\| x\| + \big(\eps + \|S_0\|\big)\, \|H_0^{\frac 12} x\|.
   \end{align*}
\end{remark}

\smallskip

\section{Semi-simplicity of the eigenvalues}  
\label{sec:reality}

In this section we establish conditions on $V$ guaranteeing that the quadratic operator polynomial $T(\la),\, \la\in \C$, induced by the formal operator sum 
$H_0-(\la-V^2)$ in Section \ref{sec:T} is strongly damped and its spectrum splits into two parts of different type. Clearly, this holds for $V=0$ since in this case $T(\la)=H_0-\la^2$ and hence $\spec(T) = \essspec(T) = (-\infty,-m] \,\dot\cup\,[m,\infty)$. For $V\ne 0$, corresponding conditions for $V$ were established in \cite{LNT06}; for bounded $V\ne 0$, weaker conditions were given in \cite{Naj83} (see also \cite{LT06}).

The notion of strongly damped operator polynomials was first introduced in \cite{MR0069030} in the finite dimensional case; in the infinite dimensional case with bounded coefficients it was elaborated in \cite{MR0169060}, \cite{MR0195291} (see also \cite{MR511976}, \cite{MR511976}), and in \cite{hhabil}; for unbounded coefficients and constant domain, it was considered in \cite{MR933487}. 

\begin{definition}
\label{strongly-damped}
The operator polynomial $T$ defined in Lemma \ref{lemma:form} is called \emph{strongly damped} if, for every $x\in \mD_T$, the quadratic polynomial $\big(T(\cdot)x,x\big)$ on $\R$ has two real and distinct roots.\\
The form polynomial $\ft$ defined in Lemma \ref{lemma:form} is called \emph{strongly damped} if, for every $x\in \mD_\ft$, the quadratic polynomial $\ft(\cdot)[x]$ on $\R$ has two real and distinct roots.
\end{definition}

\begin{remark}
\label{rem:T-ft}
\begin{itemize}
\item[{\rm i)}] If $T$ is strongly damped, then $W(T) \subset \R$.
\item[{\rm ii)}] If $\ft$ is strongly damped, then $T$ is strongly damped (since $\mD_T \subset \mD_\ft$).
\end{itemize}
\end{remark}

The following lemma and its proof generalize a result for strongly damped quadratic operator polynomials $\mL(\la) = \la^2 + \la B + C$, $\la\in \C$, which was proved in  \cite[Behauptung~5.1]{hhabil} for the case of an unbounded self-adjoint coefficient $B$ and bounded $C \ge 0$; for the case of bounded coefficients, a less direct proof may be found in \cite[Thm.\ 31.1]{MR971506}.

\begin{lemma}
\label{strongly-damped-form}
Let $\ft$ be strongly damped and denote the two different real zeros of the quadratic equation $\ft(\la)[x]=0$ by $p_-(x)<p_+(x)$ for $x\in \mD(H_0^{\frac 12})\setminus\{0\}$. If we let
\begin{align}
   \Lambda_- &:= \big\{ p_-(x) : x \in \mD(H_0^{\frac 12}) \setminus\{0\}\big\}, \quad \nu_-:= \sup \Lambda_-, \label{nu-} \\ 
   \Lambda_+ &:= \big\{ p_+(x) : x \in \mD(H_0^{\frac 12})\setminus\{0\} \big\}, \quad \nu_+:= \inf \Lambda_+,   \label{nu+}
\end{align}
then the sets $\Lambda_-$ and $\Lambda_+$ are disjoint; in particular, $\nu_- \le \nu_+$. 
\end{lemma}

\begin{proof}
Assume, to the contrary, that there exist elements $x,y \in \mD(H_0^{\frac 12})\setminus\{0\}$ with 
\begin{equation}
\label{disjoint}
  \la_0:=p_-(x) = p_+(y).
\end{equation}
Then, by the assumption on $\ft$, we have
\begin{align}
  & \ft(\la_0)[x]=\ft(\la_0)[y]=0, \label{1}\\
  & \left. \frac{{\rm d}}{{\rm d} \la} \ft(\la)[x] \right|_{\la=\la_0} > 0, \quad 
  \left. \frac{{\rm d}}{{\rm d} \la} \ft(\la)[y] \right|_{\la=\la_0} < 0. \label{2}
\end{align}
Moreover, without loss of generality, we may assume that ${\rm Re} \, \ft(\la_0)[x,y] \le 0$. Otherwise, we may replace $x$ by $-x$ since $p_-(-x)=p_-(x)$; in fact, 
\begin{equation}
\label{eq:factor}
  p_\pm(\alpha w) = p_\pm (w), \quad w\in \mD(H_0^{\frac 12}) \setminus\{0\}, \ \alpha \in \C \setminus \{0\},
\end{equation}
as $\ft(\la)[\alpha w]= |\alpha|^2 \ft(\la)[w]$ and so the two equations $\ft(\la)[\alpha w]=0$ and $\ft(\la)[w]=0$ have the same roots.

Set $z(t):= tx+(1-t)y \in \mD_\ft$, $t\in[0,1]$. First we show that $z(t) \ne 0$, $t\in [0,1]$. Clearly, $z(0)=y\ne 0$ and $z(1)=x\ne 0$. If $z(t)=0$ for some $t_0\in(0,1)$, then $x=(1-t_0)/t_0\, y$
and hence, by \eqref{eq:factor}, it follows~that
$$p_+(y)=p_-(x) = p_-((1-t_0)/t_0\, y) = p_-(y),$$ 
a contradiction to \eqref{disjoint}.

By the definition of $z(t)$ and by \eqref{1}, we see that, for all $t\in [0,1]$,
\begin{align*}
   \ft(\la_0)[z(t)] & =  t^2 \,\ft(\la_0)[x] + (1-t)^2 \, \ft(\la_0)[y] + 2 t(1-t) \, {\rm Re} \, \ft(\la_0)[x,y] \\
                    & = 2 t\,(1-t) \, {\rm Re} \, \ft(\la_0)[x,y] \le 0.
\end{align*} 
Moreover, the function
$$
  h(t):=\left. \frac{{\rm d}}{{\rm d} \la} \ft(\la)[z(t)] \right|_{\la=\la_0}, \quad t\in[0,1],
$$ 
depends continuously on $t$ and, by \eqref{2}, we have $h(0)<0$ and $h(1)>0$. Hence there exists a $t_0\in [0,1]$ such that $h(t_0)=0$. Altogether, we have
shown that the quadratic polynomial $q(\la):= \ft(\la)[z(t_0)]$, $\la\in \R$, satisfies
\[
   q(\la_0)= \ft(\la_0)[z(t_0)] \le 0, \quad q'(\la_0) = \left. \frac{{\rm d}}{{\rm d} \la} \ft(\la)[z(t_0)] \right|_{\la=\la_0} =0.
\]
As $\lim_{\la \to \pm \infty} q(\la) = -\infty$  (see the definition of $\ft$ in Lemma \ref{lemma:form}), it follows that~$q$ is non-positive and possesses at most one real zero, a contradiction to the assumption that $\ft$ is strongly damped.
\end{proof}

\begin{proposition}
\label{prop:gap0} 
The form polynomial $\,\ft$ satisfies the following implications:
\begin{itemize}
\item[{\rm i)}] if $\,\ft$ is strongly damped, then $ \ft(\la) \ge 0$ for all  $\la \in [\nu_-,\nu_+]$.
\item[{\rm ii)}]  if $\,\ft(\la_0) > 0$ for some $\la_0\in\R$, then  $\,\ft$ is strongly damped and $\nu_- \le \la_0 \le \nu_+$.
\end{itemize}
\end{proposition}

\begin{proof}
i)
Let $\la \in [\nu_-,\nu_+]$ and $x\in \mD_\ft$, $x\ne 0$, be arbitrary. Then, by the definition of $p_\pm(x)$ as the zeros of the quadratic equation $\ft(\la)[x]=0$ and by the definition of $\nu_\pm$ in \eqref{nu-}, \eqref{nu+}, we have
\begin{align}
\label{feb9}
  \ft(\la)[x] = \big( \la - p_-(x) \big) \big( p_+(x) - \la \big) \, (x,x)  \ge (\la - \nu_-)(\nu_+ - \la) \, (x,x) =: \gamma \, (x,x) \hspace*{-2mm}
\end{align}
where $\gamma\ge 0$.

ii) Let $x\in \mD_\ft$, $x\ne 0$. Since $\lim_{\la \to \pm \infty}\ft(\la)[x]  =-\infty$ and $\ft(\la_0)[x] > 0$ by assump\-tion, it follows that $\ft(\la)[x]=0$ has two real zeros $p_\pm(x)$ and $p_-(x) \!<\! \la_0 \!<\! p_+(x)$. By the definition of $\nu_\pm$,
it is immediate that $\nu_- \le \la_0 \le \nu_+$.
\end{proof}

For strongly damped quadratic operator polynomials for which at least one of the two root zones is bounded, it is well-known that the length of every Jordan chain is $1$ (see \cite[p.\ 164]{hhabil} and also \cite[Lemma~30.13]{MR971506}). The proof for two unbounded root zones is similar; we repeat it for its simplicity and for the convenience of the reader.

\begin{theorem}
\label{thm:semi-simple}
If the operator polynomial $\,T$ is strongly damped, then all eigenvalues of $\,T$ are real and semi-simple, 
i.e.\ all Jordan chains of $\,T$ have length~$1$.
\end{theorem}

\begin{proof}
Since $\,T$ is strongly damped, Proposition \ref{T-L} i) and Remark \ref{rem:T-ft} i) imply~that 
$$
  \pointspec(T) \subset W(T) \subset \R.
$$  

Assume that $\lambda_0\in \R$ is an eigenvalue of $T$ that is not semi-simple. Then, by~\eqref{Jordan}, there exist elements
$x_0,\, x_1\in\mD(T(\la_0))=\mD_T$, $x_0 \ne 0$, such that 
$$
   T(\la_0)x_0 = 0, \quad T'(\la_0)x_0 + T(\la_0)x_1 = 0.
$$
This implies that
\begin{align*}
	 & \scalar{T(\la_0)x_0}{x_0}=0, \\
	 & \left. \frac{{\rm d}}{{\rm d} \la} \scalar{T(\la)x_0}{x_0} \right|_{\la=\la_0}\!\!\!
	 = \scalar{T'(\la_0) x_0}{x_0} 
	 = -\scalar{T(\la_0)x_1}{x_0}
	 = -\scalar{x_1}{T(\la_0)x_0}
	 = 0;
\end{align*}
here we have used that $T(\la_0)$ is self-adjoint by Proposition \ref{lemma:form} i) since $\la_0$ is real. 
This shows that the quadratic polynomial $\la \mapsto \scalar{T(\la)x_0}{x_0}$ has a double zero at~$\la_0$, a contradiction to the assumption that $T$ is strongly damped.
\end{proof}

If $T$ is strongly damped, it is not immediate that the whole spectrum of $T$ is real. The reason for this is that we only have a spectral inclusion theorem for analytic operator functions with bounded coefficients (see \cite[Thm.\ 26.6]{MR971506}). Using the quadratic form polynomial $\ft$, we shall now show that $\spec(T) \subset \overline{W(T)} \subset \overline{W(\ft)} \subset \R$ if~$\nu_- < \nu_+$. 

The following definiteness properties were proved in \cite[Lemma~31.15]{MR971506} for operator polynomials with bounded coefficients (see also the original work \cite[Abschnitt~II.5.1]{hhabil}).

\begin{proposition}
\label{prop:gap} 
The form polynomial $\,\ft$ satisfies the following implications:
\begin{itemize}
\item[{\rm i)}] if $\,\ft$ is strongly damped with $\nu_- < \nu_+$, then $ \ft(\la) \gg 0$ for all $\la \in (\nu_-,\nu_+)$, and $\ft(\nu_\pm) \ge 0$.
\item[{\rm ii)}]  if $\,\ft(\la_0) \gg 0$ for some $\la_0\in\R$, then  $\,\ft$ is strongly damped with $\nu_- < \la_0 < \nu_+$. 
\end{itemize}
\end{proposition}

\begin{proof}
i) Since $\ft $ is strongly damped, Proposition \ref{prop:gap0} i) shows that $\ft(\la) \ge 0$ for all $\la\in [\nu_-,\nu_+]$ . If $\nu_- < \nu_+$ and 
$\la \in (\nu_-,\nu_+)$, then $\gamma = (\la - \nu_-)(\nu_+ - \la) > 0$ in \eqref{feb9} and hence $\ft(\la) \ge \gamma > 0$.

ii) By Proposition \ref{prop:gap0} ii), $\ft$ is strongly damped. To prove that $\nu_-<\la_0<\nu_+$, we show that for every $\gamma \in (0,\gamma_0)$ there exists an $\varepsilon>0$ such that $\ft(\la)\ge \gamma >0$ for all $\la \in (\la_0-\varepsilon,\la_0+\varepsilon)$; then Proposition~\ref{prop:gap0}~ii) implies that $\nu_- \le \la_0-\varepsilon < \la_0+\varepsilon \le \nu_+$.
Let $\gamma \in (0,\gamma_0)$ and assume, to the contrary, that there exist sequences $(\mu_n)_{n\in\N}\subset (0,\infty)$, $\mu_n\to \la_0$ for $n\to\infty$, and $(x_n)_{n\in\N} \subset \mD_\ft$, $\|x_n\|=1$, such that 
$$
 \gamma > \ft(\mu_n)[x_n], \quad  n\in\N.
$$
Set $\delta_n := 2(\mu_n-\la_0)\la_0 + (\mu_n - \la_0)^2$, $n\in \N$. Then, by~\eqref{eq:FormExpansion},
\begin{align}
\label{eq:VUnbounded}
 	 \gamma + \delta_n 
 	 &> \ft(\mu_n)[x_n] + 2\,(\mu_n-\la_0)\la_0 + (\mu_n - \la_0)^2 \nonumber \\
 	 & = \ft(\la_0)[x_n] + 2\,(\mu_n-\la_0) \,(Vx_n,x_n) \nonumber \\
 	 &\ge \ft(\la_0)[x_n] - 2\,|\mu_n - \la_0| \,\big| \scalar{Vx_n}{x_n} \big| \\
 	 &\ge \gamma_0 - 2\,|\mu_n - \la_0| \,\big| \scalar{Vx_n}{x_n} \big|. \nonumber 
\end{align} 	 
If $|(Vx_n,x_n)|$ were bounded,
the left hand side in the above inequalities \eqref{eq:VUnbounded} would tend to $\gamma$, while the right hand side would tend to $\gamma_0>\gamma$, a contradiction. Hence $|(Vx_n,x_n)| \to \infty$ for $n\to\infty$. By the Cauchy-Schwarz inequality and assumption \ref{A3}, there exist $\ra, \rb \ge 0$, $\rb<1$, such that
\[
   \big|(Vx_n,x_n) \big|^2 \le \|Vx_n\|^2 \le \ra \, \|x_n \|^2 + \rb \,\big\|H_0^{\frac{1}{2}} x_n\big\|^2, \quad n\in \N,
\]
which implies that also $\big\|H_0^{\frac{1}{2}} x_n\big\| \to \infty$ for $n\to\infty$. 

By \eqref{eq:VUnbounded}, the definition of $\ft$ in Lemma~\ref{lemma:form}, and the above inequality, it follows~that
\begin{align}
 	 \gamma + \delta_n
 	 &\ge \| H_0^{\frac{1}{2}} x_n\|^2 - \| (V-\lambda_0) x_n\|^2 - 2\, |\mu_n - \la_0| \, \big| \scalar{Vx_n}{x_n}\big| \nonumber \\
 	 \label{eq:rhsInfty}
 	 & \ge \| H_0^{\frac{1}{2}} x_n\|^2 - \| V x_n\|^2 - 2 |\la_0|\, \| V x_n\| - |\la_0|^2 - 2 \, |\mu_n - \la_0| \, \|Vx_n\|\\
   & \ge (1-\rb)\, \| H_0^{\frac{1}{2}} x_n\|^2 -\ra 
   - 2 ( |\la_0| + |\mu_n - \la_0| )\, \big( \sqrt{\ra} + \sqrt{\rb} \,\| H_0^{\frac{1}{2}} x_n\| \big) - \la_0^2. \hspace{-6mm} \nonumber
\end{align}
Since $\rb<1$ and $\| H_0^{\frac{1}{2}} x_n\|\to\infty$ for $n\to\infty$, the right hand side of the inequalities \eqref{eq:rhsInfty} tends to $\infty$, whereas the left hand side tends to $\gamma$, a contradiction.
%
\end{proof}

%


The proof of the following theorem relies implicitly on the Langer factorisation theorem on quadratic operator polynomials (see \cite[Abschnitt II.3]{hhabil} or \cite{MR0229072}) which is the main ingredient for the two lemmas from \cite{MR933487} that we use.

\begin{theorem}
\label{thm:specincl}
If $\,\ft$ is strongly damped and $\nu_- < \nu_+$, then 
\begin{align}
\label{eq:gap}
 \spec(T)  \subset (-\infty,\nu_-] \,\dot\cup \, [\nu_+,\infty) \subset \R
\end{align}
and $\nu_\pm \in \spec(T)$; if $\nu_\pm \in W(\ft)$, then $\nu_\pm \in \pointspec(T)$.
\end{theorem}

\begin{proof}
By Proposition \ref{T-L} i), we have $W(T) \subset W(\ft) \subset (-\infty,\nu_-] \,\dot\cup \, [\nu_+,\infty) \subset \R$. Hence $T$ satisfies the assumptions of 
\cite[Lemmas~1.1, 1.2]{MR933487}, which yield that $\spec(T) \subset \R$.

By Proposition \ref{prop:gap} i), we have $\ft(\la) \gg 0$ for all $\la\in (\nu_-,\nu_+)$. Since $\mD(T(\la)) = \mD_T \subset \mD_\ft=\mD(H_0^{\frac 12})$, this implies that also $T(\la) \gg 0$ for all $\la\in (\nu_-,\nu_+)$. Since $T(\la)$ is selfadjoint, it follows that $0\in \rho(T(\la))$ and hence $\la\in \rho(T)$ for all $\la\in (\nu_-,\nu_+)$, which proves \eqref{eq:gap}.

By \eqref{eq:gap} and Proposition \ref{T-L} ii) and iii), it follows that $\pointspec(L)=\pointspec(T)$ and $\spec(L)\cap \R=\spec(T)$. 
By Proposition~\ref{T-L}~i), we have $W(\ft)=W(L)$. From \cite[Behauptung~5.1]{hhabil}, it follows that the boundary points $\nu_\pm$ of $W(\ft)=W(L)$ belong to $\spec(L)$ and hence $\nu_\pm \in \spec(L) \cap \R = \spec(T)$.

If $\nu_- \!\in W(\ft)=W(L)$, then there is an $x_-\!\in \CH \setminus \{0\}$ such that $\big(L(\nu_-)x_-,x_-\big)=0$. By \eqref{L-ft} and Proposition~\ref{prop:gap}~i), we have 
$(L(\nu_-)x,x) = \ft(\nu_-)[H_0^{-\frac 12}x,H_0^{-\frac 12}x] \ge 0$ for all $x\in \CH$. 
Now the Cauchy-Schwarz inequality for the positive semi-definite inner product $(L(\nu_-)\,\cdot\,,\,\cdot\,)$ shows that 
\[
  \big|(L(\nu_-)x_-,y)\big| \le \big|(L(\nu_-)x_-,x_-)\big|^{\frac 12} \,\big|(L(\nu_-)y,y)\big|^{\frac 12} = 0
\]
for arbitrary $y \in \CH$ and hence $L(\nu_-)x_- = 0$, i.e.\ $\nu_-\in \pointspec(L)$. By Proposition~\ref{T-L}~ii), we know that $\pointspec(L)=\pointspec(T)$ which completes the proof of the last claim.
%
\end{proof}


\section{Criteria for real spectrum}
\label{sec:condV}

In this section we establish conditions on the operator $V$ guaranteeing that the abstract Klein-Gordon pencil $T$ is strongly damped and $\nu_-<\nu_+$, thus ensuring that $T$ has real spectrum $\sigma(T) \subset (-\infty,\nu_-] \,\dot\cup \, [\nu_+,\infty)$.
They generalize the conditions given in \cite{Naj83} for bounded $V$ and they weaken the conditions given in \cite{LT06} in the case when $V$ is definite, i.e.\ $V \ge 0$ or $V\le 0$.   

\begin{lemma}
\label{berlin}
Assume that {\rm \ref{A1}} and  {\rm \ref{A3}} hold, and let $\rsa$, $\rsb \ge 0$, $\rsb<1$, be the constants according to \eqref{eq:H1hBound}
Then, for $\la\in\R$, the form $\ft(\la)$ is semi-bounded with
\begin{enumerate}
\item[{\rm i)}] $\ft(\la) \ge m^2 - \big( \rsa + \rsb m + |\la|\big)^2$; \vspace{1mm}
\item[{\rm ii)}] if $\,\rsa+\rsb m < m$, then $\ft(\la) \gg 0$ for $|\la|< m - (\rsa+ \rsb m)$; 
\item[{\rm iii)}] if $\,\|VH_0^{-\frac 12}\|<1$, then $ \ft(\la) \gg 0$ for $|\la|< m - \|VH_0^{-\frac 12}\|\,m$.
\end{enumerate}
\end{lemma}

\begin{proof}
i) Let $\la \in \R$ and $x\in\mD_\ft=\mD(H_0^{\frac 12})$. By \eqref{eq:H1hBound} and Remark \ref{rem:kato}, we have, for arbitrary $\eps>0$,
\begin{align*}
\ft(\la)[x]  &= \|H_0^{\frac 12} x\|^2 - \|(V-\la) x\|^2 \\
             & \ge \big(1-(1+\eps)\rsb^2\big) \|H_0^{\frac 12} x\|^2 - (1+\eps^{-1})(\rsa + |\la|)^2 \|x\|^2\\
             & \ge \Big( \big(1-(1+\eps)\rsb^2\big) m^2 - (1+\eps^{-1})(\rsa + |\la|)^2 \Big) \|x\|^2 =: h(\eps) \|x\|^2.
\end{align*}
It is not difficult to check that the function $h:(0,\infty)\to \R$ has a maximum at $\eps_0 = (a+|\la|)/bm$ and hence
\[
  \ft(\la) \ge h(\eps_0) =  m^2 - \big( \rsa + \rsb m + |\la|\big)^2.
\]

ii) is immediate from i).

iii) Let $\la \in \R$ and $x\in\mD_\ft=\mD(H_0^{\frac 12})$.
Since $H_0 \ge m^2$, we have $\|H_0^{-\frac 12}\| \le 1/m$, $\|H_0^{\frac 12}x\|\ge m\|x\|$, and hence the estimate
\begin{align}
\ft(\la)[x]  &= \|H_0^{\frac 12} x\|^2 - \big\|\big(VH_0^{-\frac 12} -\la H_0^{-\frac 12}\big) H_0^{\frac 12}x \big\|^2 \nonumber \\
             &\ge  \Big( 1 - \Big(  \big\|VH_0^{-\frac 12}\big\| + \frac{|\la|}{m} \Big)^2 \Big) \|H_0^{\frac 12}x\|^2 \label{hot} \\
             & \ge \Big( m^2 - \big( \|VH_0^{-\frac 12}\| m + |\la| \big)^2 \Big) \|x\|^2; \nonumber
\end{align}
here, for the last estimate, we have used that the first factor is $> 0$ if (and only if) $|\la|< m - \|VH_0^{-\frac 12}\|\,m$. 
\end{proof}

\begin{remark}
Claim iii) of Lemma \ref{berlin} is stronger than claim ii) since \eqref{eq:normS} shows that $\|VH_0^{-\frac 12}\|m \le a+bm$. However, the estimate \eqref{hot} in this case does not permit to derive a lower bound for $\ft(\la)$ as in i) for all $\la\in\R$.
\end{remark}

Part i) of the following theorem was proved in \cite[Lemma~5.1]{LNT06} by means of the operator polynomial $L$ (using the inclusions $W(T) \subset W(L)$, $\spec(T) \subset \spec(L)$, see Proposition \ref{T-L} i) and ii)); parts ii) and iii) were proved in \cite{Naj83} for bounded~$V$.

\begin{theorem}
\label{real-spec}
Assume that {\rm \ref{A1}} holds, i.e.\ $\mD(H_0^{\frac 12}) \subset \mD(V)$.
\begin{enumerate}
\item[{\rm i)}] 
\label{item:nosign}
If $\,\|VH_0^{-\frac 12}\|<1$, then
\[
  \nu_- \le - m +  \|V H_0^{-\frac 12}\| \,m < 0 < m  -  \|V H_0^{-\frac 12}\| \,m \le \nu_+
\]
and hence 
$$ 
   \spec(T) 
            \subset \big( - \infty,  - m  + \|V H_0^{-\frac 12}\| \, m \big] \,\dot\cup \, \big[ m - \|V H_0^{-\frac 12}\| \, m,\infty\big).
$$    
\item[{\rm ii)}] 
\label{item:signpos}
If $\,V$ is self-adjoint with $\,V \ge 0$ and $\|VH_0^{-\frac 12}\|<2$, then
\[
  \nu_- \le - m  + \|V H_0^{-\frac 12}\| \,m < m \le \nu_+
\]
and hence 
$$ 
   \spec(T) 
            \subset \big( - \infty,  - m + \|V H_0^{-\frac 12}\| \, m \big] \,\dot\cup \, \big[ m,\infty\big).
$$    
\item[{\rm iii)}] 
\label{item:signneg}
If $\,V$ is self-adjoint with $\,V \le 0$ and $\|VH_0^{-\frac 12}\|<2$, then
\[
  \nu_- \le - m  < m - \|V H_0^{-\frac 12}\| \,m \le \nu_+
\]
and hence 
$$ 
   \spec(T) 
            \subset \big( - \infty,  - m \big]  \,\dot\cup \, \big[ m - \|V H_0^{-\frac 12}\| \, m ,\infty\big).
$$    
\end{enumerate}
\end{theorem}

\begin{proof}
i) 
%
The estimates for $\nu_\pm$ are immediate from Lemma \ref{berlin} iii) and Proposition~\ref{prop:gap}~ii).
Together with Theorem~\ref{thm:specincl}, the inclusion for $\spec(T)$ follows. 

ii) The condition $\|V H_0^{-\frac 12}\| <2$ implies that $- m  + \|V H_0^{-\frac 12}\| \,m <  m$. Then, for arbitrary 
$\la \in \big(\!- m  + \|V H_0^{-\frac 12}\| \,m,  m \big)$ and $x\in \mD_\ft = \mD(H_0^{\frac 12} )$, $\|x\|=1$, we have
\begin{align*}
 \ft(\la)[x] & 
 = \big( H_0^{\frac 12} x, H_0^{\frac 12} x \big)
 - \big( (V-\la+m)\,x, (V-\la-m)\,x \big) - m^2 \|x\|^2\\
             & \ge - \big( (V+(m-\la))\,x, (V-m-\la) \,x \big)
\end{align*}
because $H_0 \ge m^2$. 
Since $\la < m$ and $V \ge 0$, we have $V+m-\la > 0$. 
Moreover, for arbitrary $y \in \mD(V)$, we can estimate
\begin{align*}
  \big( (V-m-\la)\,y,y \big) & \le |(Vy,y)| - (m + \la) (y,y) \\
                               & \le \Big( \|V H_0^{-\frac 12}\| \, \|H_0^{\frac 12} y\| - (m + \la) \| H_0^{-\frac 12}\| \,\| H_0^{\frac 12} y\| \Big) \|y\| \\
                               & \le \Big( \|V H_0^{-\frac 12}\| - 1 -  \frac{\la}m \Big)  \|H_0^{\frac 12} y\| \, \|y\|
                               \, \le \, 0
\end{align*}
because $\la > - m  + \|V H_0^{-\frac 12}\| \,m$. Thus, as $V$ is self-adjoint, the square roots of the non-negative operators $V+m-\la$, $- (V-m-\la)$ exist and we arrive at the~estimate
\[
  \ft(\la)[x] \ge \!\big( \big( V+ m-\la \big)^{\frac 12} \big(-(V-m-\la) \big)^{\frac 12} x, \big( V+m-\la \big)^{\frac 12} \big(-(V-m-\la) \big)^{\frac 12} x \big) 
  > 0.
\]
Now Proposition \ref{prop:gap0} ii) applied twice with $\la_0 = - m  + \|V H_0^{-\frac 12}\| \,m$ and $\la_0 = m$ implies that $\nu_- \le - m  + \|V H_0^{-\frac 12}\| \,m < m \le \nu_+$. 

\smallskip

iii) The proof of iii) is completely analogous to the proof of ii).
\end{proof}

\begin{corollary}
\label{feb28}
Suppose that assumptions {\rm \ref{A1}} and  {\rm \ref{A3}} hold, let $\ra$, $\rb \ge 0$, $\rb <1$ and $\rsa$, $\rsb \ge 0$, $\rsb<1$, 
be the constants according to \eqref{eq:H1hBoundSquare} and \eqref{eq:H1hBound}, respectively. Define 
\[
  \delta_1 := \rsa+\rsb \,m, \quad \delta_2:=\sqrt\ra + \sqrt \rb \,m, \quad \delta_3 := \sqrt{\ra+\rb \,m^2},
\]
and let $i\in\{1,2,3\}$. 
\begin{enumerate}
\item[{\rm i)}] If $\delta_i< m$, then
\[
 \spec(T) \subset \big( - \infty, - m + \delta_i \big] \,\dot\cup \, \big[ m - \delta_i, \infty \big).
\]
\item[{\rm ii)}]
If $\,V$ is self-adjoint with $\,V \ge 0$ and $\delta_i< 2m$, then
\[
 \spec(T) \subset \big( - \infty, - m + \delta_i \big] \,\dot\cup \, \big[ m , \infty \big).
\] 
\item[{\rm iii)}]
If $\,V$ is self-adjoint with $\,V \le 0$ and $\delta_i< 2m$, then
\[
 \spec(T) \subset \big( - \infty, - m \big] \,\dot\cup \, \big[ m - \delta_i, \infty \big).
\]
\end{enumerate}
\end{corollary}

\begin{proof}
The claims are immediate from  Proposition \ref{real-spec} since $\|VH_0^{-\frac 12}\| \le \delta_i m$ due to the estimates in \eqref{eq:normS}.
\end{proof}

\begin{remark} 
Note that $\delta_3 \le \delta_2$ so that the assumption $\delta_3<m$ is weaker than $\delta_2<m$. Hence if $\delta_2<m$, then we have the spectral inclusions for $i=2$ and $i=3$ and thus
\[
 \spec(T) \subset \big( - \infty, - m + \delta_3 \big] \,\dot\cup \, \big[ m - \delta_3, \infty \big) \subset
 \big( - \infty, - m + \delta_2 \big] \,\dot\cup \, \big[ m - \delta_2, \infty \big).
\]
\end{remark}

\begin{remark}
\label{rem:shift}
If $\widetilde V$ is a symmetric operator in $\CH$ of the form $\widetilde V = V + c$ where $V$ satisfies 
{\rm \ref{A1}} and  {\rm \ref{A3}}, then all the above results as well as Theorem \ref{thm:main2} below also apply to $\widetilde V$; in this case, we only have to replace the spectral parameter $\la$ by the new spectral parameter $\widetilde\la=\la-c$ and $\nu_\pm$ by $\widetilde \nu_\pm := \nu_\pm -c$.
\end{remark}

\section{Simplicity of the ground states}
\label{sec:positivity}

In this section we consider the particular case of the Hilbert space $\CH=\Ltwo(\R^n)$, the self-adjoint operator $H_0=-\Delta+m^2$, and a (real-valued) multiplication operator $V$ therein. We assume that $V$ satisfies the assumptions of the previous sections so that the corresponding Klein-Gordon pencil 
$T(\la) = H_0 \dot- V^2 +2\la V - \la^2$, $\la \in \C$, is strongly damped and $\nu_-<\nu_+$. 

We shall show that if, in this case, $\nu_-$ or $\nu_+$ are eigenvalues of the Klein-Gordon pencil $T$, then they are not only semi-simple, but even simple with strictly positive eigenfunction.

To this end, we need the concept of positivity preserving and improving linear operators, and a few other definitions.

\begin{definition}
Let $(M,\, \rd\mu)$ be a measure space.
\begin{enumerate}
\item[{\rm i)}] A function $f\in\Ltwo(M,\, \rd\mu)$ is called \emph{strictly positive} if $f(x)> 0$ for almost all $x\in M$ and \emph{positive} if $f(x)\ge 0$ for almost all $x\in M$ and $f\not\equiv 0$.
\item[{\rm ii)}] A linear operator $B$ in $\Ltwo(M, \rd\mu)$ is called \emph{positivity preserving} if $Bf$ is positive for all positive $f\in\mD(B)$ and  \emph{positivity improving} if $Bf$ is strictly positive for all positive $f\in\Ltwo(M,\,\rd\mu)$.
\item[{\rm iii)}] A \emph{configuration projection} is a projection that is a multiplication operator in $\Ltwo(M, \rd\mu)$;
the range of such a configuration projection is called a \emph{configuration subspace} of $\Ltwo(M, \rd\mu)$.
\item[{iv)}] A bounded linear operator $B$ in $\Ltwo(M, \rd\mu)$ is called \emph{indecomposable} if no non-trivial configuration subspace is left invariant under $B$.
\end{enumerate}
\end{definition}

 
The simplicity of the ground states of Schr\"odinger operators is usually proved by means of the following Krein-Rutman type theorem for bounded self-adjoint operators in $\Ltwo(M, \rd\mu)$, applied to the resolvent or to the corresponding semi-group. The following two theorems were proved in \cite{faris}.

\begin{theorem}[\protect{\cite[Thm.\ 10.3]{faris}}] 
\label{thm:PF}
Let $B$ be a bounded self-adjoint operator in $\Ltwo(M,\, \rd\mu)$ which is positivity preserving
and indecomposable. If $b_+:=\max \spec(B)$ is an eigenvalue of $B$,
then $b_+$ is simple with strictly positive eigenfunction.
\end{theorem}

\begin{theorem}[\protect{\cite[Thm.\ 10.5]{faris}}] 
   \label{thm:indecomposable}
   Assume that $T_0$ is a non-negative self\-ad\-joint operator in $\Ltwo(M,\, \rd\mu)$
   and let $v$ be a real-valued measurable function on $M$
   such that the corresponding multiplication operator $V$ is $T_0$-form-bounded
   with relative bound $<1$. Denote by $T$ the operator form sum  $T_0+V$, which is 
   self-adjoint and bounded from below. Then
   
   \begin{enumerate}

      \item[{\rm i)}]
      if $\,T_0$ is indecomposable, then $T$ is indecomposable; \vspace{1mm}

      \item[{\rm ii)}]
      if $\,T_0$ is indecomposable and $(T_0 + c)^{-1}$ is positivity preserving
      for all $c>0$ and $t_-:=\min \spec(T)$ is an eigenvalue of $T$, then $t_-$ is simple with
      strictly positive eigenfunction.

   \end{enumerate}

\end{theorem}

In the sequel we apply these results to the Klein-Gordon equation with vanishing vector potential $A$ (i.e.\ $A_j\equiv 0$, $j=1,\dots,n$); here the following lemma will be~used. 

\begin{lemma}\label{lemma:pospres}
   The operator $H_0=-\Delta+m^2$ in $\Ltwo(\R^n)$ is indecomposable and $(H_0 + c)^{-1}$ is positivity preserving for all $c>0$.
\end{lemma}

\begin{proof}
   In \cite{faris} it was shown that $H_0$ is indecomposable. By the Trotter product formula (see \cite[Thm.\ VIII.31]{RSI}), we have
   \begin{align*}
      \e^{-t(-\Delta + m^2)}\, =\,
      {\rm s-\!}\lim_{n\to\infty} \bigl( \e^{t\Delta/n} \e^{-tm^2/n)} \bigr)^n.
   \end{align*}
   Since $\e^{t\Delta}$ is positivity pre\-serving (see \cite[XIII.12, Ex.\ 1]{RSIV})
   and this property is preserved under strong limits
   (see \cite{RSII}), it follows that $\e^{-tH_0}$ is positivity preserving. 
   Now the well-known formula (see~\cite[(X.98)]{RSII})
   \begin{align*}
      (H_0+c)^{-1}x\, =\, 
      \int_0^\infty \e^{-ct} \e^{-tH_0} x\,\rd t,
      \quad x\in\Ltwo(\R^n),
      &
   \end{align*}
shows that $(H_0+c)^{-1}$ is positivity preserving as well. 
\end{proof}

\begin{theorem}\label{thm:main2}
   Let $H_0=-\Delta+m^2$ and let $V$ be a symmetric operator in $\Ltwo(\R^n)$ with $\mD(H_0^{\frac 12}) \subset \mD(V)$ and $\|VH_0^{-\frac 12}\|<1$. Suppose that 
   $\nu_- < \nu_+$ for $\nu_\pm$ defined in \eqref{nu+}, \eqref{nu-}.
   If $\,\nu_+$ $(\nu_-$, respectively$)$ is an eigenvalue of $\,T$, then it is simple with strictly positive eigenfunction.
\end{theorem}

\begin{proof}
   Assume that $\nu_-$ is an eigenvalue of $T$; the proof for $\nu_+$ is analogous.
   Then $0$ is an eigenvalue of $T(\nu_-)$ and $T(\nu_-)\ge 0$ by Proposition \ref{prop:gap} i).
   By Theorem~\ref{thm:indecomposable}, $0$ is a simple eigenvalue of $T(\nu_-)$ and there is a strictly positive eigenfunction of $T(\nu_-)$ 
   at $0$ which is simultaneously an eigenfunction of $T$ at $\nu_-$.
   Since $\nu_-$ is a semi-simple eigenvalue of $T$ by Theorem~\ref{thm:semi-simple}, the assertion is proved.
\end{proof}

\begin{remark}
For non-vanishing magnetic vector potential $A$, the extremal
eigenvalues cannot be
expected to be simple in general. This phenomenon already manifests
itself
in the situation of two spatial dimensions and a constant magnetic field
$B$ (in which
case we can choose $A=x\wedge B$). Then 
\[
   H_0 =\sum_{j=1}^2 \Big(-\im \frac{\partial}{\partial x_j}-eA_j \Big)^2+m^2
\]
is the so-called Landau Hamiltonian (shifted by the constant $m^2$),
for which the spectral resolution is explicitly known and whose eigenvalues are
infinitely degenerate (see \cite{Landau30}).
Another related result is the Aharonov-Casher theorem
(see \cite[Thm.\ 6.5]{CFKS87}) which asserts that the degeneracies
of the eigenvalues of the Landau Hamiltonian
with a bounded compactly supported magnetic field are proportional to
the total flux of the magnetic field.
It is also known (see \cite{SondWils} and \cite{SimFuncInt})
that, for $A\neq 0$, the operator $H_0$ does not generate a
positivity-preserving semigroup and
thus does not satisfy the assumptions of the Krein-Rutman
theory we have employed in this section.
\end{remark}

\section{Examples of potentials}
\label{sec:examples}

In this final section we apply the above simplicity results for the ground states to the spectral problem associated with the Klein-Gordon equation \eqref{eq:KGtime} in $\R^n$ with $n\ge 3$. We show, in particular, how the assumptions of our abstract theorems work out for certain classes of scalar potentials, including Coulomb-like and Rollnik potentials.

The first class of scalar potentials was motivated by a remark in \cite{Naj83}, although the potentials therein were, in general, assumed to be bounded. In the sequel, we denote by $\hs^1(\R^n)$ the first order Sobolev space associated with $\Ltwo(\R^n)$. 

\begin{theorem}
\label{prop:najman-example}
Let $A\equiv 0$ and suppose that the potential~$q$ satisfies
\[
  q_- := \esssup_{x\in\R^n} \left\{ eq(x) - \sqrt{m^2 + \frac{\gamma^2}{\|x\|^2}} \right\} < 
         \essinf_{x\in\R^n} \left\{ eq(x) + \sqrt{m^2 + \frac{\gamma^2}{\|x\|^2}} \right\} =: q_+
\]
with $n \ge 3$ and $0 \le \gamma < (n-2)/2$. Then the form polynomial $\ft\!$ given by
\[
  \ft(\la) [\Psi] = \big( \!-\im \nabla \,\Psi,  -\im \nabla \,\Psi \big)  + m^2 \| \Psi\|^2 - \big( (eq - \la) \Psi, (eq - \overline{\la}) \Psi \big), 
  \quad \Psi\!\in\!\hs^1(\R^n), 
\]
is strongly damped and the boundary points $\nu_\pm$ of its numerical range defined in \eqref{nu-}, \eqref{nu+} satisfy 
\[
 \nu_- \le q_- < q_+ \le \nu_+.
\]
Hence the operator polynomial $T$ associated with the Klein-Gordon equation \eqref{eq:KGtime} is strongly damped, its spectrum $\sigma(T)$ is real,
\[
  \sigma(T) \subset (-\infty,q_-] \,\dot\cup\, [q_+,\infty),
\]
$\nu_\pm \!\in\! \sigma(T)$, 
all eigenvalues of $\,T$ are semi-simple, 
and if $\,\nu_-$ or $\nu_+$ are eigenvalues of~$\,T$, they are simple with strictly positive eigenfunctions.
\end{theorem}

\begin{proof}
The claims follow from Theorems \ref{thm:semi-simple} and \ref{thm:main2} if we show that 
$V=eq$ satisfies conditions \ref{A1} and \ref{A3} and that $\ft(\la) \gg 0$ for $\la \in (q_-,q_+)$. So let $\la \in (q_-,q_+)$. 
Then, by the definition of $q_\pm$ and by assumption, there exists an $\varepsilon > 0$ such that 
\begin{equation}
\label{eq:eq}
   \big| eq(x) - \la \big|  \le \sqrt{m^2 + \frac{\gamma^2}{\|x\|^2}} - \varepsilon \quad \text{for almost all\ } \ x \in \R^n.
\end{equation}
Together with Hardy's inequality (see e.g.\ \cite[Sect.\ 3.3]{MR1313735}), we thus obtain, for $\Psi\in\mD(H_0^{\frac{1}{2}}) = \hs^1(\R^n)$, 
\begin{align}
	 \big\| ( V -\la ) \Psi \big\|^2 
	 & = \int_{\R^n} \bigl| \big( eq(x) - \la \big) \Psi(x) \bigr|^2\, \rd x \nonumber \\
	 & \le 
	 \int_{\R^n} \left(m^2 + \frac{\gamma^2}{\|x\|^2} - \varepsilon^2 \right) \left| \Psi(x) \right|^2\, \rd x \nonumber 
	 \\
	 & = 
	 \big( m^2 - \varepsilon^2 \big) \|\Psi\|^2
	 +
	 \gamma^2 \, \int_{\R^n}  \frac{1}{\|x\|^2 } | \Psi(x) |^2\, \rd x \nonumber 
	 \\
	 &\le
	 \big( m^2 - \varepsilon^2 \big) \|\Psi\|^2
	 + \frac{4 \gamma^2}{(n-2)^2} \int_{\R^n} | \nabla\Psi(x) |^2\, \rd x \nonumber 
	 \\
	 &=
	 \big(  m^2 - \varepsilon^2 \big) \|\Psi\|^2 +  \frac{4 \gamma^2}{(n-2)^2}\| \nabla \Psi \|^2 \label{eq:hardy-last}  \\
	 & \le \Big( \Big( 1 - \frac{4 \gamma^2}{(n-2)^2} \Big) m^2 - \varepsilon^2 \Big) \|\Psi\|^2 
	 + \frac{4 \gamma^2}{(n-2)^2} \big\| H_0^{\frac{1}{2}} \Psi \big\|^2. \nonumber
\end{align}
This shows that $\mD(H_0^{\frac{1}{2}}) = \hs^1(\R^n) \subset \mD(V)$, i.e.\ condition \ref{A1} holds and, since $\gamma\!<\!1/2$ and $n\ge 3$, 
condition \ref{A3} holds with some $\alpha \ge 0$ and $\beta = 4 \gamma^2/{(n-2)^2}\!<1$.

The inequality \eqref{eq:hardy-last} yields that
\[
  \big\| ( V -\la ) \Psi \big\|^2 \le \big( m^2 - \varepsilon^2 \big) \|\Psi\|^2 + \| \nabla \Psi \|^2
                                  = - \varepsilon^2 \|\Psi\|^2 + \big\| H_0^{\frac{1}{2}} \Psi \big\|^2
\]
for $\la \in (q_-,q_+)$ and $\Psi\in\mD(H_0^{\frac{1}{2}}) = \hs^1(\R^n)$
and hence
\[
 \ft(\la)[\Psi] = \| H_0^{\frac{1}{2}} \Psi \|^2 - \big\| ( V -\la ) \Psi \big\|^2 \ge \varepsilon^2 \|\Psi\|^2,
\]
i.e.\ $\ft(\la) \gg 0$ for $\la \in (q_-,q_+)$, as required. 
\end{proof}

\begin{remark}
The assumptions on $q$ in \cite{Naj83} are slightly different; it is only required that $q$ satisfies a pointwise estimate as in Theorem \ref{prop:najman-example} with $\gamma=1/2$. It is not clear why, under this weaker assumption, still $\|VH_0^{-\frac 12}\|<1$ as claimed in~\cite{Naj83}.
\end{remark}

A particular case of Theorem \ref{prop:najman-example} is the Coulomb potential in $\R^n$; in the case $n=3$, the explicitly known formulas for the eigenvalues 
and eigenfunctions (see \cite{Ves83} and also \cite[Sect.\,V]{LT06}) confirm our results.

\begin{corollary}[Coulomb potential in $\R^n$]
\label{cor:coulomb}
Let $A\equiv 0$ and let $q(x) = -Ze /\|x\|$, $x\in \R^n\setminus\{0\}$, where $Z$ is the nuclear charge.    
If $Ze^2< (n-2)/2$, then the Klein-Gordon pencil $T$ is strongly damped~with 
\[
  \nu_- \le -m < 0 \le \nu_+,
\]
the spectrum of $T$ is real, 
\[
  \pointspec(T) \subset \spec(T) \subset (-\infty,-m] \,\dot\cup\, [0,\infty), 
\]
$\nu_\pm \in \spec(T)$, and all eigenvalues of $T$ are semi-simple. If $\nu_+ < m$, then $\nu_+ \in \pointspec(T)$ and $\nu_+$ is simple with 
strictly positive eigenfunction.
\end{corollary}

\begin{proof}
By Theorem \ref{prop:najman-example} (or by directly applying Hardy's inequality), we see that 
$V=eq$ is $(-\Delta+m^2)^{\frac12}$-bounded,
\begin{align*}
   \mD(H_0^{\frac{1}{2}})\subset \mD(V),
   \quad
   \| V H_0^{-\frac{1}{2}} \| \le  \frac{2Ze^2}{n-2}.
\end{align*}
Hence all claims but the last follow from Theorem \ref{prop:najman-example} if we show that $q_-\le -m$ and $q_+\ge 0$. 
These two estimates follow from the inequalities
\[
  \frac{-Ze^2}{\|x\|} - \sqrt{m^2 + \frac{(Ze^2)^2}{\|x\|^2}} \le -m, \quad
  \frac{-Ze^2}{\|x\|} + \sqrt{m^2 + \frac{(Ze^2)^2}{\|x\|^2}} \ge 0, \quad x\in \R^n\setminus\{0\}.
\]
Notice that the bounds $\nu_-\le -m$, $\nu_+\ge 0$ also follow from Theorem \ref{real-spec} iii) since the Coulomb potential is negative.  

In order to prove the last claim, we observe that $V=eq$ is $H_0$-compact where $H_0=-\Delta+m^2$ (see \cite[Lemma V.5.8]{kato}).
This implies that the difference of the resolvents of $T(\la)$ and of $H_0-\la^2$ is compact for every $\la\in\R$ and 
\[
  \essspec(T) = \big\{ \lambda \in\C : \lambda^2\in\essspec(H_0) \big\} = 
                \big(\!-\!\infty,-m\big]\,\dot\cup\,\big[m,\infty\big)
\]
(comp.\ \cite{MR0462350}). Since $\nu_\pm \!\in\! \spec(T)$ by Theorem \ref{thm:specincl}, 
we conclude that, if $0\le \nu_+\!<\!m$, then $\nu_+ \notin \essspec(T)$ or, equivalently $0 \notin \essspec(T(\nu_+))$. 
Since $T(\nu_+)$ is selfadjoint, it follows that $0\in \spec(T(\nu_+)) \setminus \essspec(T(\nu_+)) \subset \pointspec(T(\nu_+))$. 
Thus $\nu_+\in \pointspec(T)$ is simple with strictly positive eigenfunction by Theorem \ref{thm:main2}. 
\end{proof}

\begin{example}[Coulomb potential in $\R^3$]
The above results agree with the explicit formulas for the eigenvalues and eigenfunctions of the Klein-Gordon problem in $\R^3$. 
In fact, if $|Ze^2|< 1/2$, the eigenvalues are given \vspace{-1mm} by 
\begin{align*}
   & \la_{k,l} = m \left( 1+ \frac{(Ze^2)^2}{\left( k\!-\!l\!-\!\frac 12 \!+\! \sqrt{(l\!+\!\frac 12)^2 
   \!-\! (Ze^2)^2}\right)^2} \right)^{\!\!\!\!-\frac 12}\!\!\!\!\!, \ \ k=1,2,\dots,\ l=0,1,\dots, k\!-\!1.
\end{align*}
All eigenvalues $\la_{k,l}$ lie in the interval $(0,m)$, they are semi-simple, and
$\la_{k,l}$ has (geometric and algebraic) multiplicity $2l+1$. A corresponding set $\big\{ \Psi_{k,l,j} : j=0, \pm 1, \pm 2, \dots \pm l \big\}$ of linearly independent eigenfunctions is given by 
\begin{align*}
   \Psi_{k,l,j}(x) =
   N_{k,l}\, Y_{l,j}(\theta,\phi) \,
   \beta_{k,l}^{\mu_{l}  - \frac{1}{2}}\,
   \|x\|^{\mu_{l}  - \frac{1}{2}}\,
   \e^{-\beta_{k,l}\|x\|/2}
   {}_1F_1(l+1-k, 2\mu_{l} +1; \beta_{k,l}\|x\|)
\end{align*}
for $x \in \R^3$ in spherical coordinates $(\|x\|, \theta,\phi)$ and $k=1,2,\dots$, $l=0,1,2,\dots$, $j=0, \pm 1, \pm 2, \dots \pm l$; here $N_{k,l}>0$ is a normalisation constant,
$Y_{l,j}$ are the spherical harmonics, $\beta_{k,l} := 2\sqrt{m^2 - \lambda_{k,l}^2}$, 
$\mu_{l} := \sqrt{ (l+\frac{1}{2})^2 - (Ze^2)^2}$, and ${}_1F_1$ are the confluent hypergeometric functions.

In particular, the smallest eigenvalue $\la_{1,0}$ in the gap $(-m,m)$ of the essential spectrum is simple; $\la_{1,0}$ and the corresponding eigenfunction $\Psi_{1,0,0}$ are given by (see \cite[Exercise 1.11]{greiner}, \cite{Ves83} or \vspace{1mm} \cite[Sect.~5,~Ex.~1]{LT06}) 
\begin{align*}
   \la_{1,0} &= m \left( 1+ \frac{(Ze^2)^2}{\left(\frac 12 + \sqrt{\frac 14 - (Ze^2)^2}\right)^2} \right)^{\!\!\!\!-\frac 12}\!\!\!\!\!, \quad
   \\
   \Psi_{1,0,0}(x)
   &=
   N_{1,0} \Big(2 \,\|x\| \sqrt{m^2 - \lambda_{1,0}^2} \Big)^{\sqrt{1/4- (Ze^2)^2} - \frac{1}{2}}\,
   \e^{-\|x\| \sqrt{m^2 - \lambda_{1,0}^2}}
%
   , \quad x \in \R^3.
\end{align*}
Obviously, $\Psi_{1,0,0}$ is strictly positive (note that we chose $N_{1,0}\!>\!0$).
Since $\la_{1,0} \!<\! m$ and $\la_{1,0} \!=\! p_+(\Psi_{1,0,0})$, we know that $\nu_+ \le \la_{1,0} < m$ and hence, by Corollary \ref{cor:coulomb}, $\nu_+ \in \pointspec(T)$.
This illustrates the claim of Theorem \ref{thm:main2} that $\nu_+=\la_{1,0}$ is simple with strictly positive eigenfunction $\Psi_{1,0,0}$.
\end{example}

\begin{proposition}[Rollnik potentials]
Let $n=3$ and $A\equiv 0$. The scalar potential  $V=eq$  is said to be in the Rollnik class ${\mathcal R}$ if $q:\R^3\to\C$ is a measurable function~and
\begin{align*}
   \| V\|_{\mathcal R}^2 := \int_{\R^3}  \int_{\R^3}  \frac{ | V(x) V(y) | }{ |x-y|^2 }\, \rd x\, \rd y < \infty
\end{align*}
$($see \cite{RSII}$)$. If $V^2 \in {\mathcal R}$ and $\|V^2\|_{\mathcal R} < 4\pi$, then the Klein-Gordon pencil $T$ is strongly damped with $\nu_-< \nu_+$,
its spectrum $\sigma(T)$ is real with
$$ 
  \spec(T) \subset 
  \Big( \!-\! \infty,  - m  +  \sqrt{ \frac{\|V^2\|_{\mathcal R}}{4\pi}}\, m \Big] \,\dot\cup \, \Big[ m - \sqrt{ \frac{\|V^2\|_{\mathcal R}}{4\pi}} \, m,\infty\Big),
$$    
all eigenvalues of $\,T$ are semi-simple, 
$\nu_\pm \!\in\! \sigma(T)$, and if $\nu_-$ or $\nu_+$ are eigenvalues of~$\,T$, they are simple with strictly positive eigenfunctions.
\end{proposition}

\begin{proof}
It was shown in \cite[Thm.\ X.19]{RSII} that if $V^2 \in {\mathcal R}$, then $V$ is $H_0^{\frac 12}$-bounded with $H_0^{\frac 12}$-bound 0 and
\[
  \big\| V (-\Delta+m^2)^{-\frac 12} \big\| \le \sqrt{ \frac{\|V^2\|_{\mathcal R}}{4\pi}}.
\]
Hence condition \ref{A1} holds and, if $\|V^2\|_{\mathcal R} < 4\pi$, then condition \ref{A3} is satisfied with $a=0$ and $b=\sqrt{ \|V^2\|_{\mathcal R}/(4\pi)}<1$.
Thus Theorems \ref{real-spec} i), \ref{thm:semi-simple}, and \ref{thm:main2} apply and yield all claims.
\end{proof}

While the simplicity results of Section \ref{sec:positivity} only apply for vanishing magnetic potentials $A \equiv 0$, the semi-simplicity results of Section \ref{sec:reality} also apply if $A \not\equiv 0$. Here, we only mention that the condition \ref{A1}, i.e.\ $\mD(H_0^{\frac 12}) \subset \mD(V)$, may be guaranteed by the following assumptions on $A$ and \vspace{2mm} $V$:

\begin{condlist}
\CondItem{(KG1)}\label{KG1} 
  $A = (A_j)_{j=1}^n\!\in\!\Ltwoloc(\R^n)^n$;   \vspace{-0.3mm} 
\CondItem{(KG2)}\label{KG2} 
  $H_0$ is the self-adjoint realisation of $\displaystyle\sum_{j=1}^n \Bigl(-\im\,\pdiff{x_j} - e A_j \Bigr)^2 + m^2$ in $\Ltwo(\R^n)$;  \vspace{-1mm} 
\CondItem{(KG3)}\label{KG3}   
  the multiplication operator $V$ by $eq$ in $\Ltwo(\R^n)$ satisfies 
\[
 \hs^1_A(\R^n) \!:=\! \Big\{ \Psi \!\in\! \Ltwo(\R^n)\!:\!\Big( \!-\!\im\,\pdiff{x_j} \!-\! e A_j \Big) \Psi \!\in\! \Ltwo(\R^n), j=1,\dots,n \Big\} \!\subset\! \mD(V), \hspace{-8mm}
\]
\end{condlist}  
(see e.g.\ \cite[Def.\ 7.20]{LiebLoss}). 
Note that $\hs^1_A(\R^n)=\hs^1(\R^n)$ if $A \equiv 0$. For~$A\not\equiv 0$, 
the inclusion $\hs^1_A(\R^n) \subset \hs^1(\R^n)$ need not be true; $\Psi \in \hs^1_A(\R^n)$ only implies $|\Psi| \in \hs^1(\R^n)$.
%

If {\rm \ref{KG1}}, {\rm \ref{KG2}} hold, then the quadratic form $\ft_0 := \big( H_0^{\frac 12} \,\cdot\,, H_0^{\frac 12} \cdot \big)$ with domain $\mD(\ft_0):= \mD(H_0^{\frac 12})$ is given by
\begin{align}
  & \ft_0[\Psi] = \big( (-\im\,\nabla - e A) \Psi, (-\im \nabla - e A) \Psi \big) + m^2 \| \Psi \|^2, \quad \mD(\ft_0) = \hs^1_A(\R^n).
  \nonumber  
\end{align}
%
In fact, it is easy to see that the formula for $\ft_0[\Psi]$ holds for $\Psi\in\Cinfty(\R^n)$; 
since $\Cinfty(\R^n)$ is a core of $\ft_0$ by \cite[Thm.~7.22]{LiebLoss}), it extends to $\Psi\in\mD(\ft_0)=\hs^1_A(\R^n)$.


It now remains to establish conditions on $V=eq$ guaranteeing that $V$ also satisfies the relative form-boundedness assumption \ref{A3} with respect to $H_0$ so that the operators $T(\la)= H_0\dot - V^2 + 2 \la V - \la^2$ are defined for $\la\in\C$ according to Lemma~\ref{lemma:form}.

Such conditions will, in general, depend on the particular properties of the magnetic potential $A$. 
For certain classes of vector potentials $A$, it may even be possible to define the operators $T(\la)$ for $\la \in \R$ without using 
Lemma~\ref{lemma:form}; if e.g.\ $A \!=\! (A_j)_{j=1}^n\in\Lfourloc(\R^n)^n$, $\nabla\cdot A \!\in\!\Ltwoloc(\R^n)$, and
$V^2\!\in\!\Ltwoloc(\R^n)$, 
then it can be shown that $T(\lambda)$ is essentially self-adjoint on $\Cinfty(\R^n)$ by the Lein\-felder-Simader theorem (see~\cite{LeSi81}, \cite[Thm.\ 1.15]{CFKS87}).




\end{document}